\documentclass[letterpaper, 10 pt, conference]{ieeeconf}
\IEEEoverridecommandlockouts                              

\usepackage{multicol}
\usepackage{float}
\usepackage{amsfonts}
\usepackage{comment}
\usepackage{amsmath}
\usepackage{changepage}
\usepackage{amssymb}
\usepackage{graphicx}
\usepackage{adjustbox}
\usepackage{latexsym}
\usepackage{enumerate}
\usepackage[font=small,belowskip=0pt]{caption}
\usepackage{mathtools}
\usepackage{algorithm}
\usepackage{algcompatible}
\usepackage{algpseudocode}
\usepackage{bbm}
\usepackage{xcolor}
\usepackage{bm}
\usepackage{booktabs}

\definecolor{porange}{HTML}{E77500} 
\definecolor{purple}{HTML}{A020F0} 
\algrenewcommand{\algorithmiccomment}[1]{\textcolor{porange}{\hfill// #1}}
\algnewcommand{\LineComment}[1]{\Statex \textbf{\textcolor{porange}{// #1}}}

\usepackage{url}
\usepackage{hyperref}
 \hypersetup{
     colorlinks=true,
     linkcolor=blue,
     filecolor=blue,
     citecolor=blue,      
     urlcolor=black,
     }

\usepackage[%
    style=numeric-comp,
    sorting=none,
    backend=biber,
    sortcites=true,
    doi=false,
    url=false,
    firstinits=true,
    hyperref,
    isbn=false,
    eprint=false,
    minbibnames=1,   
    maxbibnames=3,
    block=none]
    {biblatex}

\AtEveryBibitem{\clearfield{pages}}

\renewbibmacro{in:}{}

\newbibmacro{string+doiurl}[1]{
    \iffieldundef{doi}{
        \iffieldundef{url}{#1}{\href{\thefield{url}}{#1} }
    }{\href{https://dx.doi.org/\thefield{doi}}{#1}}
}
\DeclareFieldFormat{title}{\usebibmacro{string+doiurl}{\mkbibemph{#1}}}
\DeclareFieldFormat[article,incollection,techreport,inproceedings,book,misc]%
    {title}%
    {\usebibmacro{string+doiurl}{\mkbibquote{#1}}}

\usepackage{amsmath}
\usepackage{bm}

\newtheorem{theorem}{Theorem}
\newtheorem{lemma}{Lemma}    

\newtheorem{remark}{Remark}
    
\newtheorem{corollary}{Corollary}
\newtheorem{proposition}{Proposition}    
\newtheorem{assumption}{Assumption}

\usepackage{ifthen}
\newboolean{include-notes}
\newboolean{include-new}
\newboolean{include-remove}
\setboolean{include-notes}{true}
\setboolean{include-new}{false}
\setboolean{include-remove}{false}

\usepackage{xcolor}
\usepackage[normalem]{ulem}
\newcommand{\jaime}[1]{\ifthenelse{\boolean{include-notes}}{\textcolor{orange}{\textbf{Jaime:} #1}}{}}
\newcommand{\haimin}[1]{\ifthenelse{\boolean{include-notes}}{\textcolor{magenta}{\textbf{Haimin:} #1}}{}}

\newcommand{\todo}[1]{\ifthenelse{\boolean{include-notes}}{\textcolor{teal}{\textbf{TODO:} #1}}{}}
\newcommand{\remove}[1]{\ifthenelse{\boolean{include-remove}}{\textcolor{red}{\sout{#1}}}{}}
\newcommand{\new}[1]{\ifthenelse{\boolean{include-new}}{\textcolor{purple}{#1}}{#1}}

\usepackage{lipsum}

\newcommand{\real}{\operatorname{Re}}

\newcommand{\reals}{\mathbb{R}}

\newcommand{\distr}{p}
\newcommand{\prob}{P}

\DeclareMathOperator*{\expectation}{\mathbb{E}}

\newcommand{\grad}{{\nabla}}

\newcommand{\jacobian}{{J}}
\newcommand{\hessian}{{H}}

\DeclareMathOperator*{\argmin}{arg\,min}

\newcommand{\diag}{\operatorname{diag}}
\newcommand{\blkdiag}{\operatorname{blkdiag}}

\newcommand{\spec}{\operatorname{spec}}

\newcommand{\opinionDyn}{{g}}
\newcommand{\attDyn}{\opinionDyn_\attstate}
\newcommand{\opnstate}{{z}}
\newcommand{\bopnstate}{\bar{\opnstate}}
\newcommand{\dopnstate}{\delta\opnstate}
\newcommand{\ddopnstate}{\delta\dot{\opnstate}}
\newcommand{\vopnstate}{\mathbf{\opnstate}}
\newcommand{\bvopnstate}{\bar{\vopnstate}}
\newcommand{\dvopnstate}{\delta\vopnstate}
\newcommand{\ddvopnstate}{\delta\dot{\vopnstate}}
\newcommand{\attstate}{{\lambda}}
\newcommand{\battstate}{{\bar{\attstate}}}

\newcommand{\bias}{{b}}

\newcommand{\damping}{{d}}
\newcommand{\saturation}{{S}}

\newcommand{\softmax}{{\sigma}}
\newcommand{\poi}{{\operatorname{PoI}}}
\newcommand{\opnidx}{{\ell}}
\newcommand{\opnidxaux}{{p}}

\newcommand{\iagent}{{i}}
\newcommand{\nagents}{{N_a}}
\newcommand{\ntheta}{{N_{\theta}}}
\newcommand{\nthetai}{{N_{\theta^i}}}

\newcommand{\iagentaux}{{j}}

\newcommand{\state}{{x}}

\newcommand{\ctrl}{{u}}

\newcommand{\csig}{{\mathbf{u}}}

\newcommand{\cset}{{\mathcal{U}}}

\newcommand{\dyn}{{f}}

\newcommand{\valfunc}{{V}}

\newcommand{\policy}{{\pi}}

\newcommand{\orderset}{\mathcal{I}}
\newcommand{\ordersetagent}{\orderset_{a}}

\newcommand{\ordersetthetai}{\orderset_{\theta_i}}

\newcommand{\ilqrValQuad}{Z}
\newcommand{\ilqrValLinear}{\zeta}

\bibliography{references}

\begin{document}

\title{\Large \bf Emergent Coordination through Game-Induced Nonlinear Opinion Dynamics}
\author{Haimin Hu$^1$, Kensuke Nakamura$^2$, Kai-Chieh Hsu$^1$, Naomi Ehrich Leonard$^2$, and Jaime Fernández Fisac$^1$
\thanks{$^1$Department of Electrical and Computer Engineering, Princeton University, {\tt\small \{haiminh,kaichieh,jfisac\}@princeton.edu}}
\thanks{$^2$Department of Mechanical and Aerospace Engineering, Princeton University, {\tt\small \{k.nakamura,naomi\}@princeton.edu}}
}
\maketitle

\begin{abstract}
We present a multi-agent decision-making framework for the emergent coordination of autonomous agents whose intents are initially undecided. Dynamic non-cooperative games have been used to encode multi-agent interaction, but ambiguity arising from factors such as goal preference or the presence of multiple equilibria may lead to coordination issues, ranging from the ``freezing robot'' problem to unsafe behavior in safety-critical events. The recently developed nonlinear opinion dynamics (NOD)~\cite{bizyaeva2022nonlinear} provide guarantees for breaking deadlocks. However, choosing the appropriate model parameters automatically in general multi-agent settings remains a challenge. In this paper, we first propose a novel and principled procedure for synthesizing NOD based on the value functions of dynamic games conditioned on agents’ intents. In particular, we provide for the two-player two-option case precise stability conditions for equilibria of the game-induced NOD based on the mismatch between agents' opinions and their game values. We then propose an optimization-based trajectory optimization algorithm that computes agents’ policies guided by the evolution of opinions. The efficacy of our method is illustrated with a simulated toll station coordination example.
\end{abstract}

\section{Introduction}

As deployments of multi-agent autonomous systems, such as self-driving truck fleets and drone swarms, continue to scale up, there is a pressing need to coordinate efficient interaction between controllable agents and uncontrollable agents, including humans.
While dynamic games (cf.~\cite{bacsar1998dynamic}) capture a rich class of interactive behaviors for multi-agent systems, existing game-theoretic formulations do not effectively coordinate agents when there is uncertainty in key parameters of the game, such as goal preferences~\cite{fridovich2020confidence, liu2023learning} and information structure \cite{zrnic2021leads}, and when there are multiple equilibria~\cite{peters2020inference}.
This can lead to dangerous behavior in which agents adopt policies leading to safety-critical deadlocks, sometimes known as the ``freezing robot'' problem~\cite{trautman2010unfreezing}.
The recently developed nonlinear opinion dynamics model~\cite{bizyaeva2022nonlinear} offers a principled way of describing opinion exchange among agents in multi-agent coordination, including social navigation scenarios.
In particular, this model provides theoretical guarantees for breaking deadlocks even with no prior bias over opinions.
However, choosing appropriate parameters for the model in practical applications remains an open problem.
In this work, we combine the best of both worlds by integrating differential games with nonlinear opinion dynamics to achieve efficient multi-agent coordination.

\subsection{Related Work}
Dynamic games have shown promise in addressing a wide range of multi-agent coordination scenarios, from autonomous driving~\cite{fisac2019hierarchical,sadigh2018planning,zanardi2021urban} and physical human-robot interaction~\cite{li2019differential} to smart-grid networks~\cite{zazo2016dynamic}.
While computing equilibrium solutions of dynamic games is typically challenging, contemporary tools have been created to enable linear-quadratic (LQ) approximations of intricate, non-convex games.
In~\cite{fridovich2020efficient}, the authors take advantage of derivative information of system dynamics and planning objectives to iteratively optimize agent trajectories in a dynamic game.
This has been shown to lead to efficient trajectories for multi-agent collision avoidance scenarios.
However, in the presence of multiple suitable equilibria, this method does not offer a solution to the equilibrium selection problem~\cite{peters2020inference}.
When the intents of a player's opponents are hidden, the dynamic game becomes a partially observable stochastic game~\cite{hansen2004dynamic}.
While intractable in general, such a game can be approximately solved using, for example, the QMDP approach~\cite{littman1995learning,hu2022sharp}, scenario-based planning~\cite{hu2022active}, and Quasi-Newton optimization~\cite{sadigh2018planning}.
In this work, we propose a novel trajectory planning framework that relies on opinion dynamics for handling ambiguities and deadlocks in partially observable stochastic games, while remaining computationally tractable by leveraging the LQ and QMDP approximation techniques.

In~\cite{cathcart2022opinion}, the authors used the nonlinear opinion dynamics model~\cite{bizyaeva2022nonlinear} for rapid and flexible breaking of social deadlock in a corridor passing problem.
In~\cite{park2021tuning}, the authors used the nonlinear opinion dynamics model~\cite{bizyaeva2022nonlinear} to investigate how cooperative behavior can emerge in \emph{static} games that are played repeatedly.
By jointly considering reciprocity and rationality, agents performed cooperatively despite the Nash equilibrium solution being non-cooperative.
The region of attraction to the mutually cooperative equilibrium was shown to increase as attention to social interaction increases.
However, neither of these works provided an answer for \emph{how} the model parameters of nonlinear opinion dynamics should be chosen.
A central contribution of this work is a principled and automatic procedure for constructing nonlinear opinion dynamics based on the outcomes of dynamic games.

 \subsection{Contributions and Paper Organization} 
In this paper, we leverage non-cooperative differential games and nonlinear opinion dynamics (NOD) to propose a novel trajectory planning framework for multi-agent emergent coordination in tasks about which agents are initially undecided.
Our contributions are threefold:
\begin{enumerate}
    \item We propose for the first time an automatic procedure for synthesizing NOD, in which the coupling parameters among opinions depend on the game value functions. We show how the NOD effectively captures opinion evolution driven by the physical state of the system.
    \item  We provide, for the two-player two-option case, precise stability conditions for NOD equilibria based on the mismatch between opinions and game values, which depend on the physical states. 
    \item We present a computationally efficient trajectory planning framework that computes agents’ policies guided by their evolving opinions such that coordination on tasks emerges.
\end{enumerate}

The paper is organized as follows. Sec.~\ref{sec:prelim} provides a brief summary of general-sum differential games and nonlinear opinion dynamics. In Sec.~\ref{sec:formulation}, we formulate the problem of interest as a differential game subject to parameter uncertainties and describe the construction of game-induced nonlinear opinion dynamics (GiNOD). In Sec.~\ref{sec:analysis}, we derive stability conditions for GiNOD. Sec.~\ref{sec:QMDP} presents our main algorithmic contributions for emergent coordination using GiNOD.
We show that our proposed planning approach leads to deadlock-free interactions in a toll station coordination scenario in Sec.~\ref{sec:sim}.
We conclude our work in Sec.~\ref{sec:conclusions}. 

\section{Preliminaries}
\label{sec:prelim}

\subsection{General-Sum Differential Games}
We consider an $\nagents$-player finite-horizon general-sum differential game governed by a nonlinear dynamical system:
\begin{equation}
\label{eq:dyn_sys}
\dot{\state} = \dyn(\state(t), \csig(t)),
\end{equation}
where $t \in \reals$ is the time, $\state \in \reals^{n_x}$ is the state of the system, $\csig(t):=\ctrl^{[1:\nagents]}(t)$ where $\ctrl^\iagent \in \reals^{n_{u_\iagent}}$ is the control of player $\iagent \in \ordersetagent := \{1, 2, \ldots, \nagents\}$.
We assume $\dyn$ is continuous in $t$ and continuously differentiable in $\{\state(t), \csig(t)\}$ uniformly in $t$.
Each player $\iagent$ seeks to minimize a cost functional:
\begin{equation}
\label{eq:cost_orig}
    J^\iagent\left(\policy^{[1: \nagents]} \right) := \int_0^T c^\iagent\left(\state(t), \csig(t) \right) dt,
\end{equation}
where $\policy^\iagent$ is player $\iagent$'s control policy with $\ctrl^\iagent(t) = \policy^\iagent(t, \state(t))$, $c^\iagent(\cdot)$ is the stage cost of player $\iagent$, and we assume $c^\iagent(\cdot)$ is twice differentiable in $\{\state(t), \csig(t)\}, \forall t$.


Finding equilibrium solutions for a general-sum differential game with nonlinear dynamics can be computationally prohibitive in the general case.
\textcite{fridovich2020efficient} approach this problem by finding a feedback Nash equilibrium to a \textit{local} approximation of the original game following an iterative linear-quadratic (ILQ) scheme.
This linearizes the dynamics and quadratizes the step cost function along the nominal trajectory in each iteration.
A finite-horizon continuous-time LQ game can then be constructed for which there exists a closed-form Riccati differential solution~\cite[Chapter~6]{bacsar1998dynamic}.
The resulting approximate feedback Nash equilibrium solution consists of a tuple of linear policies. Player $\iagent$'s policy is $\policy^{\iagent*}(t, \state(t)) = \bar{\ctrl}^\iagent(t) + K^\iagent(t) \delta\state(t) + \kappa^\iagent(t)$
where $\bar{\state}(t)$ is the nominal state and $\bar{\ctrl}^\iagent(t)$ is the nominal control of agent $\iagent$, $K^\iagent (t) \in \reals^{n_{\ctrl^\iagent} \times n_\state}$ are the gains, $\kappa^\iagent(t) \in \reals^{n_{\ctrl^\iagent}}$ are the affine terms, and $\delta\state(t) := \state(t) - \bar{\state}(t)$.
Given all players' approximate Nash equilibrium strategies, the \emph{game value function} of each player $i \in \ordersetagent$ can be locally approximated by a quadratic function:
$\valfunc^\iagent(\state) \approx \frac{1}{2} \delta\state^\top\ilqrValQuad^\iagent \delta\state + \delta\state^\top\ilqrValLinear^\iagent + v^\iagent(\bar{\state})$, where $v^\iagent(\bar{\state})$ is the value for agent~$i$ when following the nominal trajectory. $\ilqrValQuad^i \in \reals^{n_\state \times n_\state}$ and $\ilqrValLinear^i \in \reals^{n_\state}$ represent the quadratic and linear changes in the nominal game value due to small deviations from nominal state $\bar{\state}(t)$.

\subsection{Nonlinear Opinion Dynamics}
The nonlinear opinion dynamics (NOD)~\parencite{bizyaeva2022nonlinear}
model complex opinion-forming behaviors among multiple agents.
For a multi-agent system of $\nagents$ agents, each having $\nthetai$ opinions, the nonlinear opinion dynamics can be expressed as
\begin{equation}
\begin{aligned}
    \label{eq:opn_dyn_orig}
    \dot \opnstate^\iagent &= - d^{\iagent} \opnstate^{\iagent} + \bias^\iagent + \attstate^\iagent \saturation^\iagent_\opnstate (\opnstate^{\iagent}) \\
    \dot \attstate^\iagent &= -m \attstate^\iagent + \saturation^\iagent_\attstate(\vopnstate^{\iagent})
\end{aligned}
\end{equation}
where
\begin{align*}
    \saturation^{\iagent,\ell}_\opnstate (\opnstate^\iagent) =  &\saturation_1\left(
        \alpha^{\iagent} \opnstate^{\iagent}_{\opnidx} +
        \textstyle\sum_{\iagentaux \in \ordersetagent \setminus \{\iagent\} } \gamma^{\iagent \iagentaux} \opnstate^{\iagentaux}_{\opnidx}
    \right) \notag \\
    &+ \textstyle\sum_{\opnidxaux \in \ordersetthetai \setminus \{\opnidx\} } \saturation_2 \left(  
        \beta^{\iagent} \opnstate^{\iagent}_{\opnidxaux} +
        \textstyle\sum_{\iagentaux \in \ordersetagent \setminus \{\iagent\} } \delta^{\iagent \iagentaux}  \opnstate^{\iagentaux}_{\opnidxaux}
    \right).
\end{align*}
Here, $\ordersetthetai := \{1, 2, \ldots, \nthetai\}$, $\opnstate^\iagent \in \reals^{N_{\theta^i}}$ is agent $\iagent$'s opinion vector in which an element $\opnstate_{\ell}^\iagent > 0$ ($\opnstate_{\ell}^\iagent < 0$) if agent $\iagent$ favors (disfavors) option $\ell \in \mathcal{I}_{\theta^i}$, $\damping^\iagent > 0$ is the damping term, $\bias^\iagent$ represents agent $\iagent$'s own bias, $\attstate^\iagent > 0$ is the attention weight on nonlinear opinion-exchange, $\alpha^{\iagent} \geq 0$ is the self-reinforcement gain, $\beta^{\iagent} \geq 0$ is the same-agent inter-option coupling gain, $\gamma^{\iagent\iagentaux}$ is the gain of the same-option inter-agent coupling with other agent $\iagentaux$, 
$ \delta^{\iagent \iagentaux}$ is the gain of the inter-option inter-agent coupling with other agent $\iagentaux$.
$\saturation_r,~r \in \{1,2\}$, is a nonlinear saturation function satisfying $\saturation_r(0)=0,~\saturation^{\prime}_r(0)=1,~\saturation^{\prime\prime}_r(0)\neq 0,~\saturation^{\prime\prime\prime}_r(0)\neq 0$, e.g., a sigmoid function or the hyperbolic tangent function $\tanh$.

The NOD capture and enable a wide range of dynamical interactions and behaviors, notably, fast and flexible multi-agent decision-making. However, it remains an open problem to synthesize parameters for design that capitalizes on these features.
One key contribution of this paper is a novel and principled algorithmic approach for automatically synthesizing NOD from a set of differential game value functions.

\section{Modeling Indecision in Differential Games using Nonlinear Opinion Dynamics}
\label{sec:formulation}

In this section, we derive a class of nonlinear opinion dynamics, which can be automatically synthesized based on dynamic game solutions, and are later used in the emergent coordination planning framework to be presented in Sec.~\ref{sec:QMDP}.

\subsection{Differential Games with Stochastic Parameters}

In this paper, we are interested in differential games where each player's game value functions are dependent on a set of stochastic parameters, i.e. $\valfunc^i(\state) = \valfunc^i(\state; \theta^1, \ldots, \theta^{\nagents})$.
We assume that the  parameter value $\theta^i$ of each player~$i$ is supported on a discrete set, i.e. $\theta^i \in \Theta^i := \{\theta^i_1,\ldots,\theta^i_{\nthetai}\}$, which is known to all other players.
We further allow  $\Theta^\iagent$ to be \emph{heterogeneous} for different players.
Consequently, parameter 
$\theta^\iagent \sim \distr(\theta^i) := (\prob(\theta^\iagent_1),\ldots,\prob(\theta^\iagent_{\nthetai}))$ where $\distr(\theta^i)$ is a categorical distribution over the $(N_{\theta^\iagent}-1)$-simplex.
The parametrized game value function can capture a broad class of differential games with \emph{categorically different} outcomes.
We provide three examples:
\begin{enumerate}
    \item \emph{Tracking objectives.} We can encode different tracking objectives in players' stage costs.
    Specifically, we assume player~$i$'s stage cost $c^i(\cdot)$ in~\eqref{eq:cost_orig} can be decomposed into two parts: $c^\iagent\left(\cdot, \cdot; \theta^\iagent \right) := c_{I}^\iagent\left(\cdot, \cdot\right) + c_{D}^\iagent\left(\cdot, \cdot; \theta^\iagent\right)$ where $c_{I}^\iagent\left(\cdot, \cdot\right)$ is the parameter-independent part that captures, e.g., regular control objectives and safety specifications, and $c_{D}^\iagent\left(\cdot, \cdot; \theta^\iagent\right)$ depends on the parameter $\theta^\iagent$, which encodes the player's tracking objectives.
    We present a running example below for illustration.

    \item \emph{Multiple distinct equilibria.} It is common for multiple equilibria to arise in a differential game~\cite{peters2020inference}.
    In this case, let $\Theta^\iagent \equiv \Theta$ be homogeneous for all players $i \in \ordersetagent$ and each $\theta_{\ell} \in \Theta$ represent a particular equilibrium solution of the game.
    For example, $\theta_1$ represents the equilibrium solution in which player 1 yields to player 2, and $\theta_2$ in which player 2 yields to player 1. 
    
    \item \emph{Information structure.} Similar to the multiple equilibria case, different information structures in a leader-follower (Stackelberg) game~\cite{fisac2019hierarchical} can also define the parameter set $\Theta$ and affect the topology of players' trajectories. 
    For example, $\theta_1$ represents a game where player 1 is the leader and $\theta_2$ where player 2 is the leader.
\end{enumerate}

\noindent \textbf{Running example:}
Consider two autonomous vehicles proceeding towards a toll station as depicted in Fig.~\ref{fig:toll}.
The toll booths are modeled as static obstacles that the cars shall avoid.
Additional safety-critical specifications include avoiding collisions with other vehicles and driving off the road.
For each  vehicle $i \in \{1,2\}$, let $\Theta^i = \{\theta^i_1, \theta^i_2\}$, where $\theta^\iagent_1$ ($\theta^\iagent_2$) represents the option that vehicle $i$ goes through toll booth 1 (toll booth 2).
The parameter-independent stage cost $c_{I}^\iagent\left(\cdot, \cdot\right)$ captures velocity tracking, safety specifications, and fuel consumption.
The parameter-dependent stage cost $c_{D}^\iagent\left(\state^i; \theta^\iagent\right) := -w^i_{\theta_i} \mathbf{1}[\state^i \in \mathcal{T}_{\theta^i}]$, where $w^i_{\theta_i} > 0$ is the weight and $\mathbf{1}[\cdot]$ is the indicator function, produces a reward (negative cost) when vehicle $i$ is inside the light green \emph{target region} $\mathcal{T}_{\theta^i}$, and zero reward otherwise.

\begin{figure}[!hbtp]
  \centering
  \includegraphics[width=1.0\columnwidth]{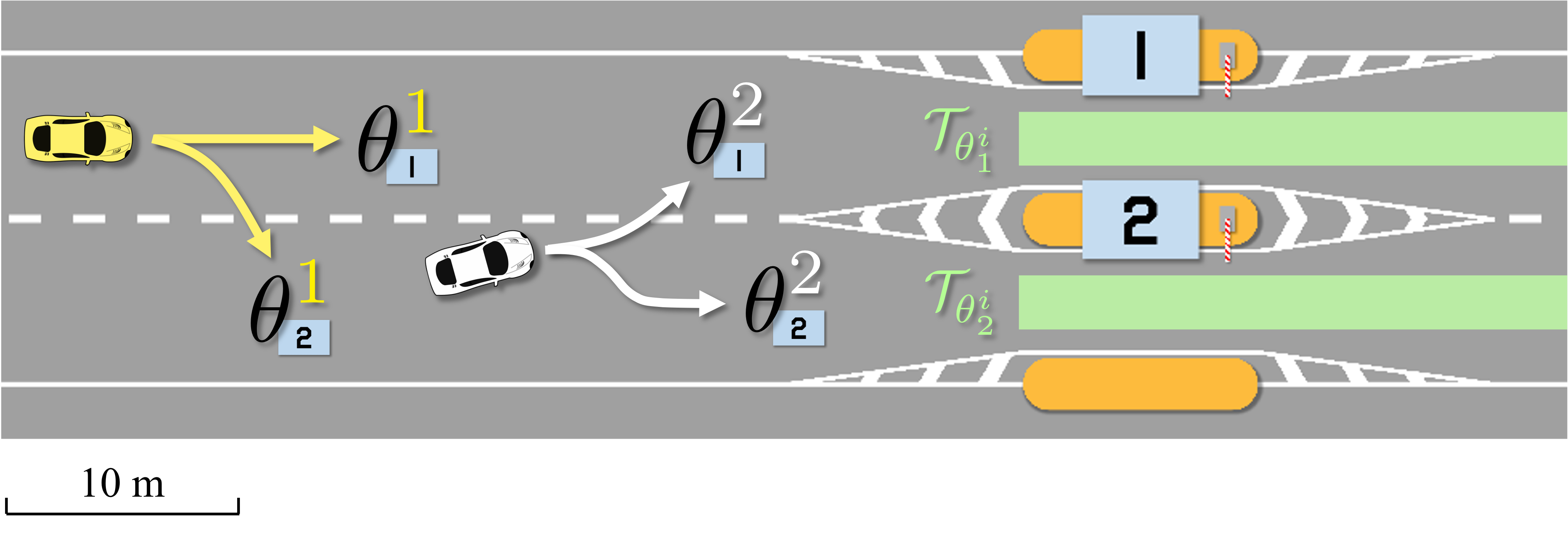}
  \caption{\label{fig:toll} Emergent coordination involving two autonomous cars at a toll station.
  Superscripts of parameters $\theta$ denote the agent number and subscripts denote the agent's preferred toll booth.
  }
\end{figure}

\begin{remark}
The above game formulation is related to the partially observable stochastic game (POSG)~\cite{hansen2004dynamic} in that opponent agents' parameters are uncertain.
Our formulation \emph{additionally} models the ego agent's parameter as a random variable, which represents the ego agent's \emph{indecision}.
\end{remark}

\subsection{QMDP Approximation and Subgames}
\label{sec:formulation:QMDP}

In this paper, we adopt the QMDP approximation technique~\cite{littman1995learning} in the game-theoretic setting for tractable computation of players' strategies.
Each player~$i$'s game value function under parameter uncertainty is computed by solving a QMDP planning problem:
\begin{equation}
\label{eq:approx_vfunc_prob}
\begin{aligned}
    &\tilde{\valfunc}^{i}(\state; \distr(\theta^1), \ldots, \distr(\theta^\nagents)) := \textstyle\min_{\ctrl^i \in \cset^i} c_I^i(\state, \ctrl^{[1:\nagents]}) + \\
    &\quad~\ \expectation_{\theta^i \sim \distr(\theta^i),~\forall i \in \ordersetagent} \left[ \valfunc^\iagent(\state^+; \theta^1, \ldots, \theta^{\nagents}) \right] \\
    &\state^+ = \bar{\dyn} \left(\state, \ctrl^i, \{\policy^j(\state; \theta^1, \ldots, \theta^{\nagents}) \}_{j \in \ordersetagent \setminus \{i\}} \right),
\end{aligned}
\end{equation}
where set $\cset^i \subseteq \reals^{n_{u_\iagent}}$ encodes the control limit of agent~$i$, $c_{I}^\iagent\left(\cdot, \cdot\right)$ is the parameter-independent stage cost, game value function $\valfunc^\iagent(\state; \theta^1, \ldots, \theta^{\nagents})$ results from equilibrium policies under a set of \emph{fully revealed} parameters, $\bar{\dyn}(\cdot)$ is the time-discretized dynamics~\eqref{eq:dyn_sys}, and $\policy^j(\state; \theta^1, \ldots, \theta^{\nagents})$ is player $j$'s equilibrium policy of the game parametrized by $(\theta^1, \ldots, \theta^{\nagents})$.
We call the game governed by a particular set of parameters $(\theta^1, \ldots, \theta^{\nagents})$ a \emph{subgame}.
In this paper, we solve for approximate subgame value functions using the ILQ Game method~\cite{fridovich2020efficient}.
QMDP~\eqref{eq:approx_vfunc_prob} optimistically assumes that the parameter uncertainties disappear in one step after \emph{the ego agent} takes an action, and that all opponents are \emph{clairvoyant} playing their corresponding subgame policies.

\begin{assumption}
We assume the availability of Subroutine $\mathcal{S}$ that computes (approximate) subgame value functions $\valfunc^i(\state; \theta^1, \ldots, \theta^{\nagents})$ for all players $i \in \ordersetagent$ and all possible parameter combinations $\theta^1 \in \Theta^1,\ldots,\theta^{\nagents} \in \Theta^{\nagents}$.
\end{assumption}

QMDP~\eqref{eq:approx_vfunc_prob} is now tractable as a single-agent trajectory optimization problem, and can be solved repeatedly in a receding horizon fashion.
Nonetheless, propagating parameter uncertainties $\distr(\theta^i)$ can still be challenging.
Existing works such as~\cite{peters2020inference,hu2022active} use Bayesian inference~\cite{chen2003bayesian} to propagate $\distr(\theta^i)$ based on observations of opponents' behaviors.
Drawbacks of Bayesian approaches include intractability for general (multi-modal) distributions and difficulty in defining an appropriate likelihood model.
In this paper, we take an alternative approach by interpreting $\distr(\theta^i)$ as the agent's \emph{degree of indecision}, and model its time evolution using nonlinear opinion dynamics.

\noindent \textbf{Running example:}
There are four subgames resulting from the parameter combinations $(\theta^1_1,\theta^2_1)$, $(\theta^1_1,\theta^2_2)$, $(\theta^1_2,\theta^2_1)$, and $(\theta^1_2,\theta^2_2)$, which encode agents' choices to go through a particular toll booth.
We solve each subgame using the ILQ method~\cite{fridovich2020efficient} for agents' state trajectories that correspond to an approximate local Nash equilibrium solution.

\subsection{From Probabilities to Opinions}
We let opinion state vector $\opnstate^i \in \reals^{N_{\theta^i}}$ model distribution $\distr(\theta^i)$, leveraging the softmax operation:
\begin{equation}
\label{softmax}
    \prob(\theta^i_\ell) \gets \softmax_\ell(\opnstate^i) := \frac{e^{\opnstate^i_\ell}}{\sum_{q=1}^{N_{\theta^i}} e^{\opnstate^i_q}},
\end{equation}
and define $\softmax(\opnstate^i) := \left(\softmax_1(\opnstate^i), \ldots, \softmax_{N_{\theta^i}}(\opnstate^i)\right)$.

By delegating uncertainties to opinions, given a physical state $\state$ and all players' opinions $\vopnstate := (\opnstate^1, \ldots, \opnstate^\nagents)$, we can define player~$i$'s \emph{opinion-weighted} game value function:
\begin{equation}
\begin{aligned}
\label{eq:approx_vfunc_opn}
    \hat{\valfunc}^{i}&(\vopnstate, \state) := \\
    &\sum_{\ell_1=1}^{N_{\theta^1}} \cdots \sum_{\ell_{\nagents}=1}^{N_{\theta^{\nagents}}} \left(\prod_{i=1}^\nagents \softmax_{\ell_i}(\opnstate^i)\right) \valfunc^i(\state; \theta^1_{\ell_1}, \ldots, \theta^{\nagents}_{\ell_{\nagents}}),
\end{aligned}
\end{equation}
as a proxy for $\expectation_{\theta^1, \ldots, \theta^{\nagents}} \left[ \valfunc^\iagent(\state; \theta^1, \ldots, \theta^{\nagents}) \right]$, the expected game value in QMDP~\eqref{eq:approx_vfunc_prob}.

\begin{remark}
The neutral opinion $\opnstate = (0,\ldots,0)$ corresponds to an (uninformative) uniform distribution $\distr(\theta) \gets \softmax(\opnstate) = \frac{1}{\ntheta} \mathbf{1}_{\ntheta}$, where $\mathbf{1}_{\ntheta} \in \reals^{\ntheta}$ is a vector of all ones.
\end{remark}

\subsection{Synthesizing Opinion Dynamics from Subgames}
\label{sec:formulation:GiNOD}

\noindent \textbf{Opinion evolution as gradient flow.}
In this section, we propose a \emph{constructive way} to synthesize opinion dynamics from subgames.
We model the agents as selfish players who seek to drive their opinions in a direction that minimizes their expected game value.
To this end, we let all agents implement the gradient flow~\cite{hochreiter2001gradient} dynamics (the continuous-time counterpart of gradient descent) that describe the evolution of their opinion states:
\begin{equation}
\label{eq:grad_flow}
    \dot{\vopnstate} = \begin{bmatrix}
    -\grad_{\opnstate^i} \hat{\valfunc}^{i}(\vopnstate, \state)
    \end{bmatrix}_{i \in \ordersetagent}
\end{equation}
where $\left[\cdot\right]_{i \in \ordersetagent}$ denotes vector concatenation by rows.

\noindent \textbf{Linear opinion dynamics.}
The local behavior of gradient flow~\eqref{eq:grad_flow} around a given $\bar{\vopnstate}$ induces a linear opinion dynamics model, originally introduced as a weighted-averaging process 
in~\cite{Abelson1964}.
This is given by linearizing~\eqref{eq:grad_flow} at $\bar{\vopnstate}$:
\begin{equation}
\label{eq:lin_opn_dyn_orig}
    \ddvopnstate = \begin{bmatrix}
    -\hessian_{11}^1(\state) & \cdots & -\hessian_{1 \nagents}^1(\state) \\
    \vdots & \ddots & \vdots \\
    -\hessian_{\nagents 1}^{\nagents}(\state) & \cdots & -\hessian_{\nagents \nagents}^{\nagents}(\state)
    \end{bmatrix}
    \begin{bmatrix}
    \dopnstate^1 \\
    \vdots \\
    \dopnstate^{\nagents}
    \end{bmatrix},
\end{equation}
where $\ddvopnstate := \vopnstate - \bar{\vopnstate}$, $\dopnstate^i := \opnstate^i - \bar{\opnstate}^i$, and
\begin{equation*}
    \hessian_{ij}^i(\state) := \left.\grad_{\opnstate^i\opnstate^j} \hat{\valfunc}^{i}(\vopnstate, \state)\right|_{\bar{\vopnstate}}
\end{equation*}
is the Hessian matrix of $\hat{\valfunc}^{i}(\cdot)$ with respect to $\opnstate^i$ and $\opnstate^j$, evaluated at $\bar{\vopnstate}$.
We can rewrite linear opinion dynamics~\eqref{eq:lin_opn_dyn_orig} equivalently as
\begin{equation}
\begin{aligned}
\label{eq:lin_opn_dyn_sub}
    \ddopnstate^i_\ell = &\alpha^i_\ell(\state) \dopnstate^i_\ell
    + \textstyle \sum_{\substack{j \neq i \\ j=1}}^{\nagents} \gamma^{ij}_\ell(\state) \dopnstate^j_\ell
    \\
    &+\textstyle\sum_{\substack{p \neq \ell \\ p=1}}^{\nthetai} \beta^i_{\ell p}(\state) \dopnstate^i_p
    + \textstyle\sum_{\substack{p \neq \ell \\ p=1}}^{\nthetai} \sum_{\substack{j \neq i \\ j=1}}^{\nagents} \eta^{ij}_{\ell p}(\state) \dopnstate^j_p,
\end{aligned}
\end{equation}
for player~$i$'s parameter $\theta^i_\ell$, where the \emph{state-dependent} dynamics parameters are defined as
\begin{subequations}
\label{eq:dyn_param}
\begin{align}
    \alpha^i_\ell(\state) &:= -\left[ \hessian_{ii}^i(\state) \right]_{\ell\ell} \\
    \gamma^{ij}_\ell(\state) &:= -\left[ \hessian_{ij}^i(\state) \right]_{\ell\ell} \\
    \beta^i_{\ell p}(\state) &:= -\left[ \hessian_{ii}^i(\state) \right]_{\ell p} \\
    \eta^{ij}_{\ell p}(\state) &:= -\left[ \hessian_{ij}^i(\state) \right]_{\ell p}
\end{align}
\end{subequations}
Here, $\left[\hessian(\cdot)\right]_{\ell p}$ denotes the entry located at the $\ell$-th row and $p$-th column of $\hessian(\cdot)$.
The reasoning behind those dynamics parameters is the same as~\eqref{eq:opn_dyn_orig} with additional flexibility to model inter-option dependency (see Appendix A in \cite{bizyaeva2022nonlinear}).
The linearized gradient flow dynamics~\eqref{eq:lin_opn_dyn_orig} synthesized from subgames fall into the category of linear opinion dynamics~\cite{Abelson1964}, also known, in discrete time, as the DeGroot model~\cite{degroot1974reaching}.

\noindent \textbf{Game-induced nonlinear opinion dynamics.}
Motivated by the recent discovery made in~\parencite{bizyaeva2022nonlinear} that the agreement and disagreement equilibria in linear opinion dynamics are not structurally stable and require special opinion dynamic gains,
we modify linear opinion dynamics~\eqref{eq:lin_opn_dyn_orig} and propose the Game-induced Nonlinear Opinion Dynamics (GiNOD):
\begin{equation}
\label{eq:GINOD}
\begin{aligned}
    &\ddopnstate^\iagent =\opinionDyn_\opnstate^\iagent(\dvopnstate, \attstate^\iagent, \state) := -D^\iagent \dopnstate^\iagent + \attstate^\iagent \begin{bmatrix}
        \opinionDyn^\iagent_\ell(\dvopnstate, \state)
    \end{bmatrix}_{\opnidx \in \ordersetthetai}  \\
    &\opinionDyn^\iagent_\opnidx(\dvopnstate, \state) := \saturation_1 \left(\alpha^\iagent_\opnidx(\state) \dopnstate^\iagent_\opnidx
    + \textstyle \sum_{\iagentaux \in \ordersetagent \setminus \{\iagent\}} \gamma^{\iagent \iagentaux}_\opnidx(\state) \dopnstate^\iagentaux_\opnidx \right)
    + \\
    &\quad \textstyle\sum_{\opnidxaux \in \ordersetthetai \setminus \{\opnidx\}} \saturation_2 \left( \beta^\iagent_{\opnidx \opnidxaux}(\state) \dopnstate^\iagent_\opnidxaux
    + \sum_{\iagentaux \in \ordersetagent \setminus \{\iagent\}} \eta^{\iagent \iagentaux}_{\opnidx \opnidxaux}(\state) \dopnstate^\iagentaux_\opnidxaux \right),
\end{aligned}
\end{equation}
for all player $\iagent \in \ordersetagent$ and parameters $\theta^\iagent_\opnidx,~\opnidx \in \ordersetthetai$, where $D^\iagent = \operatorname{diag}(\damping^\iagent_1,\ldots,\damping^\iagent_{\nthetai}) \in \reals^{\nthetai \times \nthetai}$ is a symmetric positive definite matrix that describes damping of the opinion states, $\attstate^i \in \reals$ is the attention.
The attention $\attstate^i$ can be interpreted as a scaling factor multiplying the saturated gradient flow dynamics $\opinionDyn^i_\ell(\cdot)$. A greater $\attstate^i$ promotes opinion formation ($\opnstate^i$ departing the origin), and a smaller $\attstate^i$ discourages opinion formation ($\opnstate^i$ approaching the origin).
Note that the full opinion states are recovered by $\vopnstate = \dvopnstate + \bar{\vopnstate}$, where the nominal opinion states $\bar{\vopnstate}$ are recursively updated using an iterative algorithm (Alg.~\ref{alg:opn_game}) to be introduced in the next section.

\subsection{When to Pay Attention in Games: the Price of Indecision}
We now introduce for each player~$i$ a measure called the \emph{Price of Indecision} (PoI) defined as
\begin{equation}
\label{eq:PoI}
\begin{aligned}
    &\poi^i(\vopnstate, \state) := \\
    &\max_{\ell_j \in \mathcal{I}_{\ell_j},~\forall j \neq i} \frac{\sum_{\ell_i \in \ordersetthetai} \softmax_{\ell_i}(\opnstate^i) \valfunc^i(\state; \theta^1_{\ell_1}, \ldots, \theta^{\nagents}_{\ell_{\nagents}}) }{ \min_{p_i \in \ordersetthetai} \valfunc^i(\state; \theta^1_{\ell_1}, \ldots, \theta^i_{p_i}, \ldots, \theta^{\nagents}_{\ell_{\nagents}}) }.
\end{aligned}
\end{equation}
The PoI, inspired by the Price of Anarchy~\cite{koutsoupias2009worst}, is a ratio lower bounded by $1$ that measures how player~$i$'s efficiency degrades due to indecision.
For a set of worst-case parameters selected by opponent players $j \neq i$, PoI will be large if the (expected) game value when player~$i$ chooses not to form an opinion (the numerator in~\eqref{eq:PoI}) outweighs the game value when player~$i$ declares an optimal opinion (the denominator in~\eqref{eq:PoI}).
Given the player's PoI, we introduce a state-and-opinion-dependent dynamic equation for evolving player~$i$'s attention~\cite{bizyaeva2022nonlinear}:
\begin{equation}
\label{eq:att_dyn}
    \dot{\attstate}^i = \attDyn^i(\attstate^i, \vopnstate, x) := - m^i \attstate^i + \rho^i\left( \poi^i(\vopnstate, x) - 1 \right),
\end{equation}
where $m^i > 0$ and $\rho^i > 0$ are damping and scaling parameters.
The PoI-based dynamics~\eqref{eq:att_dyn} increase the attention $\attstate^i$ when $\poi^i$ is large, which promotes opinion formation.
In other words, the agents are more inclined to form their opinions \emph{only if}  doing so increases their efficiency.


\section{Stability Analysis of Game-Induced NOD}
\label{sec:analysis} 

In this section, we derive precise stability conditions for GiNOD equilibria~\eqref{eq:GINOD} in a two-player, two-option setting, i.e. $\ordersetagent = \{1,2\}$, $\Theta_1 =\{\theta^1_1, \theta^1_2\}$, and $\Theta_2 =\{\theta^2_1, \theta^2_2\}$.
Note again that the elements in $\Theta_1$ and $\Theta_2$ need not coincide with each other.
Key properties\footnote{All informal conclusions listed here are subject to additional technical assumptions, which can be found in the Theorems and Corollary.} of GiNOD discovered from our analysis are summarized below:
\begin{enumerate}
    \item The neutral opinion for both players is unstable for an \emph{arbitrarily} small damping parameter. That is, when $\bvopnstate = 0$,  $\dvopnstate = 0$ is an unstable equilibrium of GiNOD and thus indecision is easily broken (Theorem~\ref{thm:destab}),
    \item When both agents have formed their non-neutral opinions,  $\dvopnstate = 0$ is a locally exponentially stable equilibrium for GiNOD for an \emph{arbitrarily} small damping parameter, if both agents' opinions correspond to a lower (i.e. better) game value from the current physical state (Theorem~\ref{thm:opn_game}),
    \item If for both agents there is no difference between game values of distinct options, then the opinions are driven purely by the damping terms (Corollary~\ref{coro:opn_game}).
\end{enumerate}
We start by reviewing two useful linear algebra lemmas.

\begin{lemma}[Theorem 4.2.12 in~\cite{horn1994topics}]
\label{lem:spec_kron}
Let $A \in \reals^{m \times m}$ and $B \in \reals^{n \times n}$ with $\lambda \in \spec(A)$ and  $\mu \in \spec(B)$, where $\spec(\cdot)$ denotes spectrum. Then $\lambda \mu$ is an eigenvalue of $A \otimes B$,  where $\otimes$ is the Kronecker product.
Any eigenvalue of $A \otimes B$ arises as such a product of eigenvalues of $A$ and $B$.
\end{lemma}

\begin{lemma}
\label{lem:eval_diag}
Let $D = d I_{n}$ where $d \in \reals$ and $I_n$ is the identity matrix in $\reals^{n \times n}$.
If $\lambda$ is an eigenvalue of $H \in \reals^{n \times n}$, then $d+c\lambda$ is an eigenvalue of $D + cH$ where $c \in \reals$.
\end{lemma}

\begin{proof}
Since $\lambda \in \spec(H)$, we have that $Hv = \lambda v$ where $v$ is the eigenvector associated with $\lambda$.
It follows that $(D + cH)v = dIv + cHv = (d+c\lambda)v$.
\end{proof}

In the two-player, two-option case, it is possible to derive and analyze each entry of the system matrix in opinion dynamics~\eqref{eq:lin_opn_dyn_orig}, as shown in  Lemma~\ref{lem:H}.
To ease the notation we denote value function $\valfunc^i_{\ell p} := \valfunc^i(\state;\theta^1_\ell,\theta^2_p)$.

\begin{lemma}
\label{lem:H}
Let value function $\hat{\valfunc}^{i}(\vopnstate, \state)$ be defined in~\eqref{eq:approx_vfunc_opn}.
Matrix $\mathbf{H}(x) := -\left.\left[\jacobian_{\vopnstate} \grad_{\opnstate^i} \hat{\valfunc}^{i}(\vopnstate, \state)\right]_{i \in \ordersetagent}\right|_{\bvopnstate} = \Gamma \otimes H$ with $H := \begin{bmatrix} 1 &-1 \\ -1 &1 \end{bmatrix}$ and $\Gamma := \begin{bmatrix} a_1 &b_1  \\ b_2 &a_2 \\ \end{bmatrix}$ where
\begin{align*}
    &a_1 := \phi_a(\bopnstate^1) \left[ \softmax_1(\bopnstate^2) \left( \valfunc^1_{11} - \valfunc^1_{21} \right) + \softmax_2(\bopnstate^2) \left( \valfunc^1_{12} - \valfunc^1_{22} \right) \right] \\
    &a_2 := \phi_a(\bopnstate^2) \left[ \softmax_1(\bopnstate^1) \left( \valfunc^2_{11} - \valfunc^2_{12} \right) + \softmax_2(\bopnstate^1) \left( \valfunc^2_{21} - \valfunc^2_{22} \right) \right] \\
    &b_i := \phi_b(\bopnstate^1) \phi_b(\bopnstate^2) \left(-\valfunc^i_{11}-\valfunc^i_{22}+\valfunc^i_{12}+\valfunc^i_{21}\right) \\
    &\phi_a(\bopnstate^i) := \left(\softmax_1(\bopnstate^i) - \softmax_2(\bopnstate^i)\right) \phi_b(\bopnstate^i)\\
    &\phi_b(\bopnstate^i) := \softmax_1(\bopnstate^i) \softmax_2(\bopnstate^i)
\end{align*}
for $i \in \{1,2\}$, and $\jacobian_{\vopnstate}(\cdot)$ denotes the Jacobian matrix with respect to $\vopnstate$.
Furthermore, $\spec(\mathbf{H}(x)) = \{0, 0, a_1 + a_2 \pm \left((a_1-a_2)^2 + 4 b_1 b_2\right)^{1/2}\}$.
\end{lemma}

\begin{proof}
    Entries of $\mathbf{H}(x)$ are computed from the Jacobian matrices of $\grad_{\opnstate^i} \hat{\valfunc}^{i}(\vopnstate,\state)$ for $i \in \{1,2\}$.
    Spectrum of $\mathbf{H}(x)$ follows by Lemma~\ref{lem:spec_kron} since $\spec(\Gamma) =\{\frac{1}{2} [a_1 + a_2 \pm \left((a_1-a_2)^2 + 4 b_1 b_2\right)^{1/2}] \}$ and $\spec(H) = \{0,2\}$.
\end{proof}

It is easily observed that $\dvopnstate = 0$ is an equilibrium of~\eqref{eq:GINOD}. To facilitate the analysis of stability of $\dvopnstate = 0$, we consider GiNOD~\eqref{eq:GINOD} with a \emph{steady-state} attention $\attstate^1 = \attstate^2 \equiv \battstate := \lambda_\infty(\overline{\poi}) > 0$ under attention dynamics~\eqref{eq:att_dyn} with a fixed PoI, i.e. $\poi^1(\cdot) = \poi^2(\cdot) \equiv \overline{\poi} \geq 1$.
The steady-state attention $\battstate$ is guaranteed to exist in practice since $m^i > 0$ and thus $\lim_{t \rightarrow \infty} \attstate^i(t) < \infty$ for any $\overline{\poi} < \infty$.
We also assume the same damping term for both players' options, i.e., $D^1=D^2=D=\diag(\damping,\damping)$.
Define block diagonal matrix $\mathbf{D} = \blkdiag(D,D)$. 

\begin{theorem}[Instability at neutral opinion]
\label{thm:destab}
Let $\bvopnstate = 0$, i.e all agents hold a neutral opinion.
Then, $\dvopnstate = 0$ is an unstable equilibrium of GiNOD~\eqref{eq:GINOD}
if $\damping < 2\battstate \real(\sqrt{b_1b_2})$.
\end{theorem}

\begin{proof}
When $\bvopnstate = 0$, linearization of GiNOD at $\dvopnstate = 0$ gives linear system $\ddvopnstate = (-\mathbf{D} + \battstate \mathbf{H}) \dvopnstate$, where $a_1 = a_2 = 0$ in $\mathbf{H}$ since $\phi_a(\bopnstate^1) = \phi_a(\bopnstate^2) = 0$.
By Lemma~\ref{lem:eval_diag} and~\ref{lem:H}, we have $\spec(-\mathbf{D} + \battstate \mathbf{H}) = \{-d, -d, -\damping \pm 2 \battstate\sqrt{b_1 b_2}\}$.
Thus, $\dvopnstate = 0$ is unstable if
$d < 2\battstate\real(\sqrt{b_1b_2})$.
\end{proof}

\begin{remark}[Instability and (dis)agreement]
\label{rmk:instab}
It is desirable that, at the neutral opinion, GiNOD can be made unstable even if damping $\damping$ is not close to zero, so that agents can quickly and reliably form a non-neutral opinion to break any deadlock.
This requires that $\sqrt{b_1b_2}$ has a non-zero real part, or $b_1b_2 > 0$.
Since $\phi_b(\bopnstate^1) \phi_b(\bopnstate^2) > 0$, $b_1b_2 > 0$ if and only if $\valfunc_b := \prod_{i \in \{1,2\}}\left(-\valfunc^i_{11}-\valfunc^i_{22}+\valfunc^i_{12}+\valfunc^i_{21}\right) > 0$.
This is true in two cases.
First, if both players find that agreeing to the same option would be costly (e.g. squeezing into the same toll station in the Running Example), i.e. $\valfunc^1_{\ell\ell}$ and $\valfunc^2_{pp}$ are large for some $\ell, p \in \{1,2\}$, then $\valfunc_b > 0$ as a result of multiplying two negative numbers.
Second, if the cost of disagreement is high for both players (typical case in cooperative settings), i.e. $\valfunc^1_{\ell \neg\ell}$ and $\valfunc^2_{p\neg p}$ are large for some $\ell, p \in \{1,2\}$, where $\neg\ell$ denotes the alternative option to $\ell$, then $\valfunc_b > 0$ since it is the product of two positive numbers.
\end{remark}

\begin{theorem}[Opinion reflects game value]
\label{thm:opn_game}
If for $\ell_1, \ell_2 \in \{1,2\}$ it holds that $\softmax_{\ell_1}(\bopnstate^1) > \softmax_{\neg\ell_1}(\bopnstate^1)$, $\softmax_{\ell_2}(\bopnstate^2) > \softmax_{\neg\ell_2}(\bopnstate^2)$, $\valfunc^1_{\ell_1\ell_2} < \valfunc^1_{\neg\ell_1\ell_2}$, $\valfunc^2_{\ell_1\ell_2} < \valfunc^2_{\ell_1\neg\ell_2}$, and $a_1a_2 > b_1b_2$, then  $\dvopnstate=0$ is a locally exponentially stable equilibrium of GiNOD with an arbitrarily small damping $d > 0$.
\end{theorem}

\begin{proof}
We follow the proof of Theorem~\ref{thm:destab} by examining the spectrum of the linearized system matrix $-\mathbf{D} + \battstate\mathbf{H}$.
From $\softmax_{\ell_1}(\bopnstate^1) > \softmax_{\neg\ell_1}(\bopnstate^1)$, $\softmax_{\ell_2}(\bopnstate^2) > \softmax_{\neg\ell_2}(\bopnstate^2)$, $\valfunc^1_{\ell_1\ell_2} < \valfunc^1_{\neg\ell_1\ell_2}$, and $\valfunc^2_{\ell_1\ell_2} < \valfunc^2_{\ell_1\neg\ell_2}$, we have that $a_1, a_2 < 0$.
$\dvopnstate = 0$ is locally exponentially stable with any $d > 0$ if and only if $\max(\spec(-\mathbf{D} + \battstate\mathbf{H})) = \max\{-d, -d + \battstate (a_1 + a_2 + ((a_1-a_2)^2 + 4 b_1 b_2)^{1/2})\} < 0$.
Since the steady-state attention $\battstate >0$, if $a_1 + a_2 + ((a_1-a_2)^2 + 4 b_1 b_2)^{1/2} < 0$, then $\max(\spec(-\mathbf{D} + \battstate\mathbf{H})) < 0$ with an arbitrarily small $d > 0$.
Solving the inequality gives $a_1a_2 > b_1b_2$.
\end{proof}

\begin{remark}[Interpreting inequality $a_1a_2 > b_1b_2$]
\label{rmk:opn_game}
We observe that $a_1a_2 > b_1b_2$ holds with mild assumptions on game values when both players \emph{have formed their opinions}, that is, $\softmax_{\ell_1}(\bopnstate^1)\rightarrow 1$ and $\softmax_{\ell_2}(\bopnstate^2)\rightarrow 1$ for some $\ell_1,\ell_2 \in \{1,2\}$, and players' opinions reflect their game values (conditions $\valfunc^1_{\ell_1\ell_2} < \valfunc^1_{\neg\ell_1\ell_2}$ and $\valfunc^2_{\ell_1\ell_2} < \valfunc^2_{\ell_1\neg\ell_2}$ in Theorem~\ref{thm:opn_game}).
Since $a_1, a_2 < 0$, $a_1a_2 > b_1b_2$ trivially holds if $b_1b_2 < 0$.

When $b_1b_2 > 0$, $a_1a_2 > b_1b_2$ holds if $a_1a_2 / b_1b_2 > 1$.
Without loss of generality, assume $\softmax_1(\bopnstate^i) \rightarrow 1$ and $\softmax_2(\bopnstate^i) \rightarrow 0$ for both players.
Since $\softmax_1(\bopnstate^i) \rightarrow 1$, the ratio $a_1a_2/b_1b_2 \rightarrow c_1 c_2 \valfunc^\prime$ where $c_1 := (\softmax_1(\bopnstate^1)-\softmax_2(\bopnstate^1)) / \softmax_2(\bopnstate^2) \gg 1$, $c_2 := (\softmax_1(\bopnstate^2)-\softmax_2(\bopnstate^2)) / \softmax_2(\bopnstate^1) \gg 1$ and 
$\valfunc^\prime := \valfunc_a / \valfunc_b$ where both the numerator $\valfunc_a := \left(\valfunc^1_{11}-\valfunc^1_{21}\right) \left(\valfunc^2_{11}-\valfunc^2_{12}\right)$ and the denominator $\valfunc_b := \prod_{i \in \{1,2\}} \left(-\valfunc^i_{11}-\valfunc^i_{22}+\valfunc^i_{12}+\valfunc^i_{21}\right)$ are positive.
Therefore, as long as the game values $\valfunc^i_{\ell_1 \ell_2}$ are such that $\valfunc^\prime$ is not too small (e.g., when those game values are roughly of the same magnitude), $a_1a_2 > b_1b_2$ holds.
\end{remark}

\begin{remark}[Guidelines for choosing damping]
Based on Theorem~\ref{thm:destab} and~\ref{thm:opn_game}, it is recommended to pick a damping term that satisfies $0 < \damping < 2\battstate\sqrt{b_1b_2}$ when $b_1b_2 > 0$.
In this way the equilibrium $\dvopnstate = 0$ of GiNOD is unstable at the neutral opinion,
while remaining locally exponentially stable when agents' opinions match their game value differences.
\end{remark}

\begin{corollary}[Identical game values cannot form opinion]
\label{coro:opn_game}
If both agents have the same game values for different options, i.e. $\valfunc^1_{11}=\valfunc^1_{21}$, $\valfunc^1_{12}=\valfunc^1_{22}$, $\valfunc^2_{11}=\valfunc^2_{12}$, $\valfunc^2_{21}=\valfunc^2_{22}$, then the opinions are purely driven by the damping term $\damping$, i.e., $\dvopnstate = 0$ is always locally exponentially stable.
\end{corollary}

\begin{proof}
By Lemma~\ref{lem:H} we have in this case $\mathbf{H} = 0$.
GiNOD~\eqref{eq:GINOD} becomes $\ddopnstate^\iagent = -D^\iagent \dopnstate^\iagent$ for $i \in \{1,2\}$.
\end{proof}

\section{Emergent Coordination Planning using Game-Induced Nonlinear Opinion Dynamics}
\label{sec:QMDP}

In this paper, we seek to combine the best attributes from differential games and game-induced nonlinear opinion dynamics towards efficient, deadlock-free multi-agent emergent coordination, in which all players are initially undecided about their parameters of the game.
To this end, we first formulate an \emph{opinion-weighted} QMDP based on the up-to-date opinions evolved with GiNOD in Sec.~\ref{sec:QMDP:L0}.
Then we modify and extend the QMDP formulation using ideas from cognitive hierarchy to enable \emph{active} opinion manipulation in Sec.~\ref{sec:QMDP:L1}.
The overall framework of our approach applied in receding horizon fashion is summarized in Algorithm~\ref{alg:opn_game}.

\subsection{Opinion-Weighted QMDP}
\label{sec:QMDP:L0}

Given a physical state $\state$ and players' opinions $\vopnstate$, we can formulate the opinion-weighted QMDP by combining QMDP~\eqref{eq:approx_vfunc_prob} and opinion-weighted game value function~\eqref{eq:approx_vfunc_opn}.
player~$i$'s strategy is given by:
\begin{equation}
\label{eq:QMDP_L0}
\begin{aligned}
    &\policy^i_{\text{L0}}(\state, \vopnstate) := \textstyle\argmin_{\ctrl^i \in \cset^i} c_I^i(\state, \ctrl^i) + \hat{\valfunc}^{i}(\vopnstate, \state^+) \\
    &\state^+ = \bar{\dyn} \left(\state, \policy^i_{\text{L0}}(\state, \vopnstate), \{\policy^j(\state; \theta^1, \ldots, \theta^{\nagents}) \}_{j \in \ordersetagent \setminus \{i\}} \right),
\end{aligned}
\end{equation}
where $c_{I}^\iagent\left(\cdot, \cdot\right)$ is the parameter-independent stage cost, $\hat{\valfunc}^{i}(\cdot)$ is the opinion-weighted game value function defined in~\eqref{eq:approx_vfunc_opn}, and $\policy^j(\state; \theta^1, \ldots, \theta^{\nagents})$ is player $j$'s equilibrium policy of the subgame parametrized by $(\theta^1, \ldots, \theta^{\nagents})$.
We refer to policy~\eqref{eq:QMDP_L0} as the \emph{Level-0} opinion-weighted QMDP (L0-QMDP) policy, whose namesake will become clearer in the next section as we introduce the Level-1 QMDP policy.
Aligned with the QMDP principle, an agent using the L0-QMDP policy first declares an action, then commits to an option and assumes that the parameter uncertainties of other agents disappear.
The following proposition shows that if the subgames are LQ games, then the L0-QMDP, under mild assumptions, can be cast as a convex quadratic program (QP), which can be solved efficiently via off-the-shelf solvers.

\begin{proposition}
\label{prop:QMDP}
If all subgames are LQ games, the (physical) state evolves under a control-affine dynamic model, i.e. $\state^+ = \bar{\dyn}(\state) + \sum_{\iagent \in \ordersetagent} \bar{g}^i(\state) \ctrl^{i}$, stage cost $c^i_I(\cdot,\cdot)$ is a convex quadratic function in $\ctrl^i$, and control set $\cset^i$ is convex, then QMDP~\eqref{eq:QMDP_L0} is a convex QP.
\end{proposition}

\begin{proof}
Each term $\valfunc^i(\state^+; \theta^1_{\ell_1}, \ldots, \theta^{\nagents}_{\ell_{\nagents}})$ in $\hat{\valfunc}^{i}(\cdot)$ can be expanded by plugging in the dynamic model as
$\valfunc^i(\cdot) = [(\ctrl^i - \bar{\ctrl}^{i})^{\top} g^i(\state)^{\top} \bar{\ilqrValQuad}^i + \bar{\ilqrValLinear}^{i,\top} ]g^i(\state)(\ctrl^i - \bar{\ctrl}^{i}) + \bar{C}$, which is a quadratic function in $\ctrl^i$, where $(\bar{\ctrl}^i,\bar{\ilqrValQuad}^i, \bar{\ilqrValLinear}^i)$ are shorthand notations for the value function parameters of the subgame defined by $(\theta^1_{\ell_1}, \ldots, \theta^{\nagents}_{\ell_{\nagents}})$, and $\bar{C} \in \reals$ is a constant term that does not depend on $\ctrl^i$.
Therefore, the overall cost of player~$i$ containing a weighted sum of $\valfunc^i(\cdot)$ terms and the stage cost is also a convex quadratic function in $\ctrl^i$.
\end{proof}

Despite its simplicity and efficient computation, we note that the L0-QMDP policy does not take into account opinion evolution as a result of agents' actions, thereby unable to actively steer the opinions.
In the next section, we modify the L0-QMDP to enable active manipulation of agents' opinions.

\begin{remark}
As an alternative to QMDP~\eqref{eq:QMDP_L0} where opponents are assumed as clairvoyant players, we may solve a QMDP \emph{Game}:
\begin{align*}
    &\policy^i(\state, \vopnstate) := \textstyle\argmin_{\ctrl^i \in \cset^i} c_I^i(\state, \ctrl^i) + \hat{\valfunc}^{i}(\vopnstate, \state^+) \\
    &\state^+ = \bar{\dyn} \left(\state, \policy^1(\state, \vopnstate), \ldots, \policy^\nagents(\state, \vopnstate) \right),
\end{align*}
in which the agents' QMDP problems are coupled.
If the conditions in Prop.~\ref{prop:QMDP} are satisfied and additionally $\cset^i = \reals^{n_{u_i}}$ holds, then it can be shown that the QMDP Game is an LQ Game, whose global feedback Nash equilibrium can be computed efficiently via coupled Riccati equations~\cite{bacsar1998dynamic}.
\end{remark}

\subsection{Actively Manipulating Opinions}
\label{sec:QMDP:L1}

We now introduce the \emph{Level-1} opinion-weighted QMDP (L1-QMDP) policy, which is inspired by the established work on cognitive hierarchy ($K$-level reasoning)~\cite{stahl1994experimental}.
The ego player using the L1-QMDP policy assumes that all opponents apply the L0-QMDP policy.
This way, the ego player can declare two actions sequentially in time - the first one evolves the current opinions forward in time through GiNOD, and all uncertainties disappear after the ego's second action is determined.
Given players' physical state $\state$, opinions $\vopnstate$, and attentions $\Lambda := (\attstate^1,\ldots,\attstate^\nagents)$, we can formulate the L1-QMDP planning problem as:
\begin{equation}
\label{eq:QMDP_L1}
\begin{aligned}
     \policy^i_{\text{L1}}(\state, \vopnstate, \Lambda&) = \ctrl^i_0(\state, \vopnstate, \Lambda) \\
     \min_{\ctrl_0^i, \ctrl_1^i \in \cset^i} &c_I^i(\state_0, \ctrl^i_0) + c_I^i(\state_1, \ctrl^i_1) + \hat{\valfunc}^{i}(\vopnstate_1, \state_2) \\
     \text{s.t.} \quad & \state_0 = \state,~\vopnstate_0 = \vopnstate,~\Lambda_0 = \Lambda\\
     &\state_1 = \bar{\dyn}(\state_0, \ctrl_0^i, \{\policy^j_{\text{L0}}(\state_0, \vopnstate_0)\}_{j \in \ordersetagent \setminus \{i\}}) \\
     &\state_2 = \bar{\dyn}(\state_1, \ctrl_1^i, \{\policy^j(\state_1; \theta^1, \ldots, \theta^{\nagents}) \}_{j \in \ordersetagent \setminus \{i\}} ) \\
     &\vopnstate_1 = \bar{\opinionDyn}_\opnstate(\vopnstate_0, \Lambda_0, \state_1(\ctrl^i_0)),
\end{aligned}
\end{equation}
where $\bar{\opinionDyn}_\opnstate(\cdot) := [\bar{\opinionDyn}_\opnstate^i(\cdot)]_{i \in \ordersetagent}$ is the discrete-time joint opinion dynamics given by concatenating players' time-discretized GiNOD $\opinionDyn_\opnstate^i(\cdot)$ defined in~\eqref{eq:GINOD}.
From~\eqref{eq:QMDP_L1} we can see that the control action of the ego agent~$i$ who players the L1-QMDP policy is optimized given the knowledge that it is able to affect opinion $\vopnstate_1$ through GiNOD $\bar{\opinionDyn}_\opnstate(\cdot)$.
Problem~\eqref{eq:QMDP_L1} is in general a non-convex trajectory optimization problem as it involves the nonlinear physical dynamics $\bar{\dyn}(\cdot)$ and GiNOD $\bar{\opinionDyn}_\opnstate(\cdot)$.
Nonetheless, since the dimension of decision variables (agent~$i$'s control) is oftentimes low, \eqref{eq:QMDP_L1} can be efficiently solved by gradient-based numerical solvers.

\begin{algorithm}[!hbtp]
	\caption{Receding Horizon QMDP using GiNOD}
	\label{alg:opn_game}
	\begin{algorithmic}[1]
	\Require Initial state $\state(0)$, opinions $\vopnstate(0)$, nominal opinions $\bvopnstate(0)$, attentions $\Lambda(0)$, horizon step $T$
    
    \State Initialize time step $t \gets 0$
    \For{$t = 0,1,\ldots,T-1$}
    	\LineComment{Solve subgames}
    	\State $\valfunc^i(\state(t); \theta^1, \ldots, \theta^{\nagents}) \gets$ Solve subgames using Subroutine $\mathcal{S}$ for all players $i \in \ordersetagent$ and all parameter combinations $\theta^1 \in \Theta^1,\ldots,\theta^{\nagents} \in \Theta^{\nagents}$
    
        \LineComment{Construct opinion dynamics}
        \State Set nominal opinions $\bvopnstate(t) \gets \vopnstate(t)$
        \State Construct GiNOD $\opinionDyn_\opnstate^i(\dvopnstate, \attstate^i, \state)$ in~\eqref{eq:GINOD} and attention dynamics $\attDyn^i(\attstate^i, \vopnstate, x)$ in~\eqref{eq:att_dyn} for all players $i \in \ordersetagent$
        
        \LineComment{Compute QMDP policies}
        \State $\ctrl^i(t) \gets$ Compute control action using the L0-QMDP policy $\policy^i_{\text{L0}}(\state(t), \vopnstate(t))$ in~\eqref{eq:QMDP_L0} or the L1-QMDP policy $\policy^i_{\text{L1}}(\state(t), \vopnstate(t), \Lambda(t))$ in~\eqref{eq:QMDP_L1} for each player $i \in \ordersetagent$
    
        \LineComment{Update state, opinions, and attentions}
        \State $\state(t+1) \gets$ Integrate $\dyn(\state(t), \csig(t))$
        \If{$t \geq 1$}
            \State $\dopnstate^i(t) \gets$ Integrate GiNOD $\opinionDyn_\opnstate^i(\dvopnstate(t-1), \attstate^i(t-1), \state(t))$ for all players $i \in \ordersetagent$
            \State $\attstate^i(t) \gets$ Integrate attention dynamics $\attDyn^i(\attstate^i(t-1), \vopnstate(t), x(t))$ for all players $i \in \ordersetagent$
            \State $\vopnstate(t) \gets \bvopnstate(t-1) + \dvopnstate(t)$
        \EndIf
    
    \EndFor
	\end{algorithmic}
\end{algorithm}

\section{Simulation Results}
\label{sec:sim}
We apply the receding-horizon opinion-weighted QMDP planning framework (Alg.~\ref{alg:opn_game}) to the toll station coordination task described in the Running Example.
Both vehicles $i \in \{1,2\}$ are described by a kinematic bicycle model~\cite{zhang2020optimization}, whose state is defined as $\state^i = (p_x^i, p_y^i, \varphi^i, v^i)$.
Here, $p_x^i$ and $p_y^i$ are the center position of car $i$'s rear axes, $\varphi^i$ is the yaw angle with respect to the $x$-axis, and $v^i$ is the velocity with respect to the rear axes.
The joint state vector is $x := (x^1, x^2) \in \reals^8$.
All continuous-time dynamics were discretized with a time step of $\Delta t = 0.2$ s using the forward Euler method.
We used a JAX~\cite{bradbury2020jax}-based implementation of the ILQ Game method~\cite{fridovich2020efficient} as the Subroutine $\mathcal{S}$ for solving the subgames.
The QMDP optimization problems were modeled and solved using CasADi~\cite{andersson2019casadi}.
The open-source code is available online.\footnote{\url{https://github.com/SafeRoboticsLab/opinion_game}}


\begin{figure}[!hbtp]
  \centering
  \includegraphics[width=1.0\columnwidth]{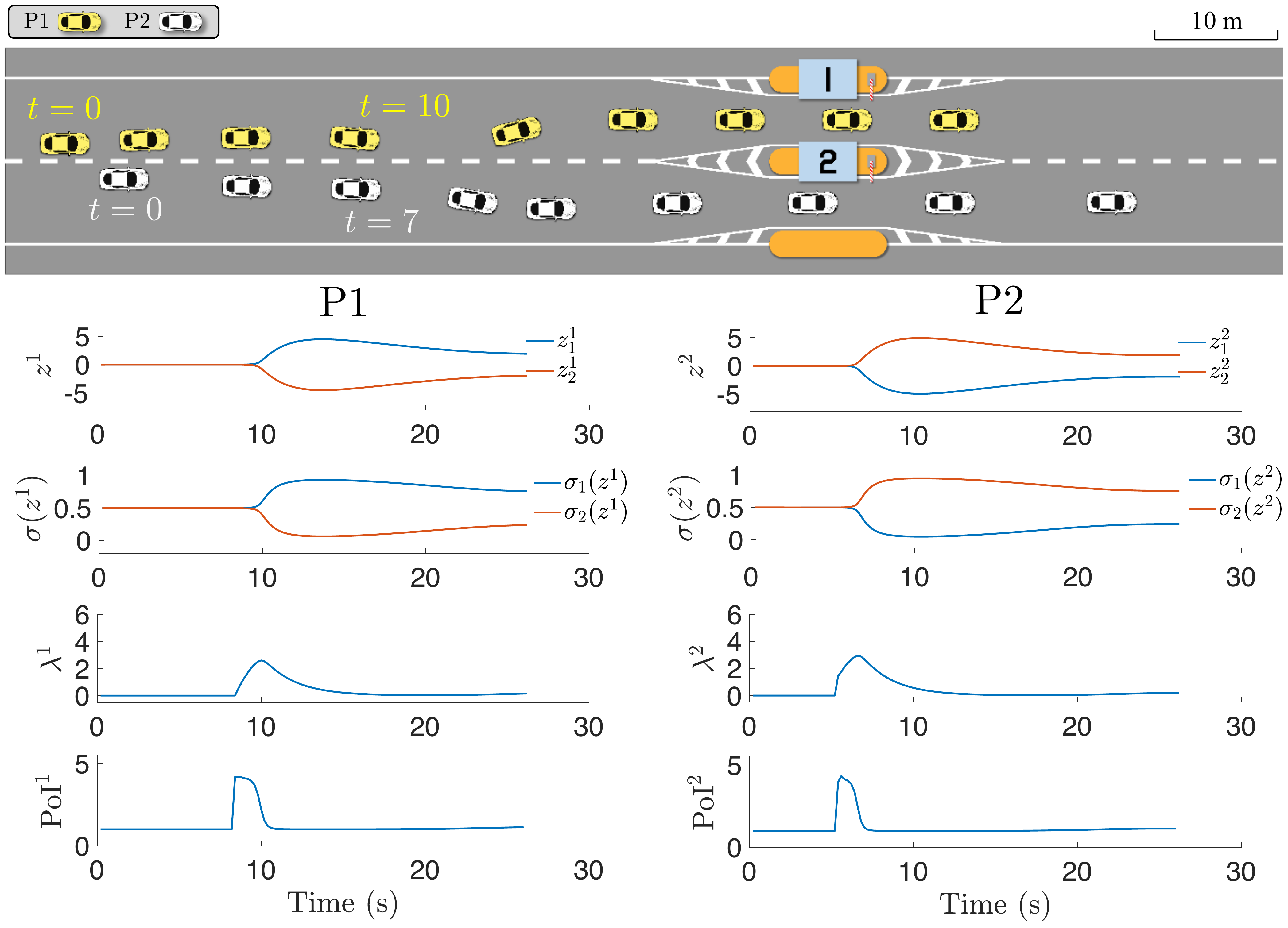}
  \caption{\label{fig:L0_t1} Agents' state, opinion (original and softmax), attention, and PoI trajectories using the L0-QMDP policy with homogeneous cost weights. Vehicle snapshots are plotted every 3 seconds.
  }
\end{figure}

The physical states of the vehicles were initialized to be $\state(0) = (0~\text{m}, 5~\text{m}, 0~\text{rad}, 3~\text{m/s}, 5~\text{m}, 2~\text{m}, 0~\text{rad}, 3~\text{m/s})$.
We modeled the \emph{initially undecided} agents by using an almost-neutral initial opinion $\opnstate^i_{\ell_i} = \epsilon$ for $i \in \{1,2\}$ and $\ell_i \in \{1,2\}$ in all simulations, where $\epsilon > 0$ is a small number that prevents opinions from staying at equilibrium $\dvopnstate = 0$ forever.
Recall that the parameter-dependent stage cost is defined as $c_{D}^\iagent\left(\state^i; \theta^\iagent\right) = -w^i_{\theta_i} \mathbf{1}[\state^i \in \mathcal{T}_{\theta^i}]$ where $w^i_{\theta^i} > 0$ is the cost weight that encodes agent~$i$'s degree of preference to go through a toll booth.

\noindent \textbf{Homogeneous cost weights.}
We first examine a case in which both vehicles had the same cost weight $w^1_1 = w^1_2 = w^2_1 = w^2_2 = 15$.
The resulting closed-loop state, opinion, attention, and PoI trajectories using the L0-QMDP policy are plotted in Fig.~\ref{fig:L0_t1}.
At the beginning of the simulation, agents' opinions were constantly neutral since the vehicles were farther away from the toll stations and the parameter-dependent cost $c_{D}^\iagent\left(\state^i; \theta^\iagent\right)$ evaluated to $0$ for both options, hence the subgame values were identical.
This validated Corollary~\ref{coro:opn_game}.
As $c_{D}^\iagent(\cdot)$ started to produce nonzero rewards for car 1 at around $t=10$ s and car 2 at around $t=7$ s, the attentions driven by the PoI spiked up, agents rapidly formed an opinion, and both cars safely passed through a toll station, which empirically verified Theorem~\ref{thm:destab}.

\noindent \textbf{Heterogeneous cost weights.}
Next, we consider a set of heterogeneous cost weights $w^1_1 = 40$, $w^1_2 = 50$, $w^2_1 = 50$, $w^2_2 = 40$, encoding that car 1 prefers to go through toll booth 2, and car 2 is more inclined to visit toll booth 1.
The state and opinion trajectories when both cars are using the \emph{L0-QMDP} policy are plotted in Fig.~\ref{fig:L0_t2}.
Due to the interference from car 2 (cutting in front of car 1), car 1 formed an opinion to stay in the left lane and went through the less preferred toll station.
In another trial under the same initial condition, we applied the \emph{L1-QMDP} policy to car 1 while keeping the L0-QMDP policy for car 2.
The resulting trajectories are shown in Fig.~\ref{fig:L1L0}.
By leveraging the active opinion manipulation feature of the L1-QMDP policy, car 1 was able to plan a more efficient trajectory towards its preferred toll booth.

\begin{figure}[!hbtp]
  \centering
  \includegraphics[width=1.0\columnwidth]{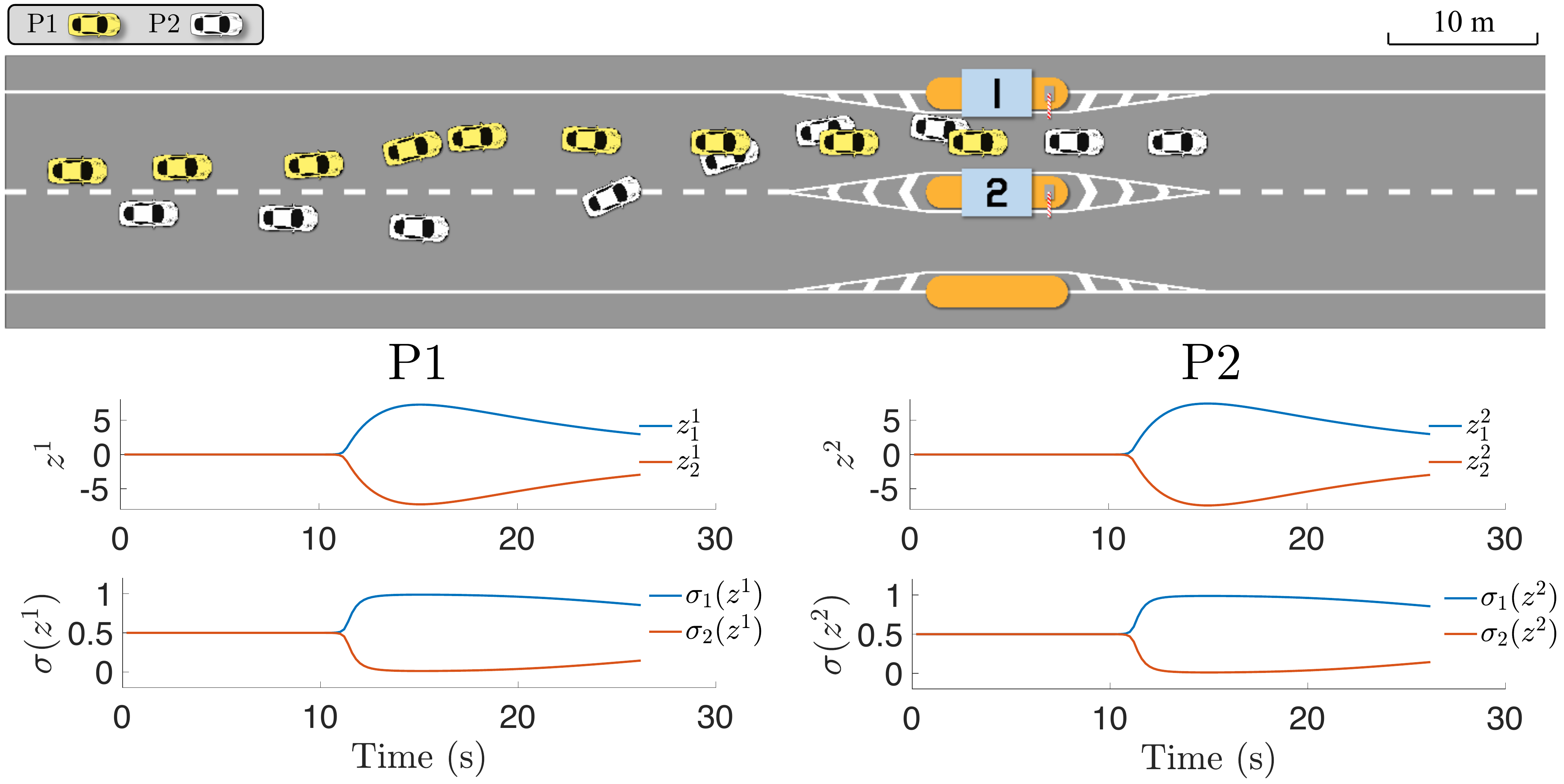}
  \caption{\label{fig:L0_t2} Agents' state and opinion (original and softmax) trajectories using the L0-QMDP policy with heterogeneous cost weights.
  }
\end{figure}

\begin{figure}[!hbtp]
  \centering
  \includegraphics[width=1.0\columnwidth]{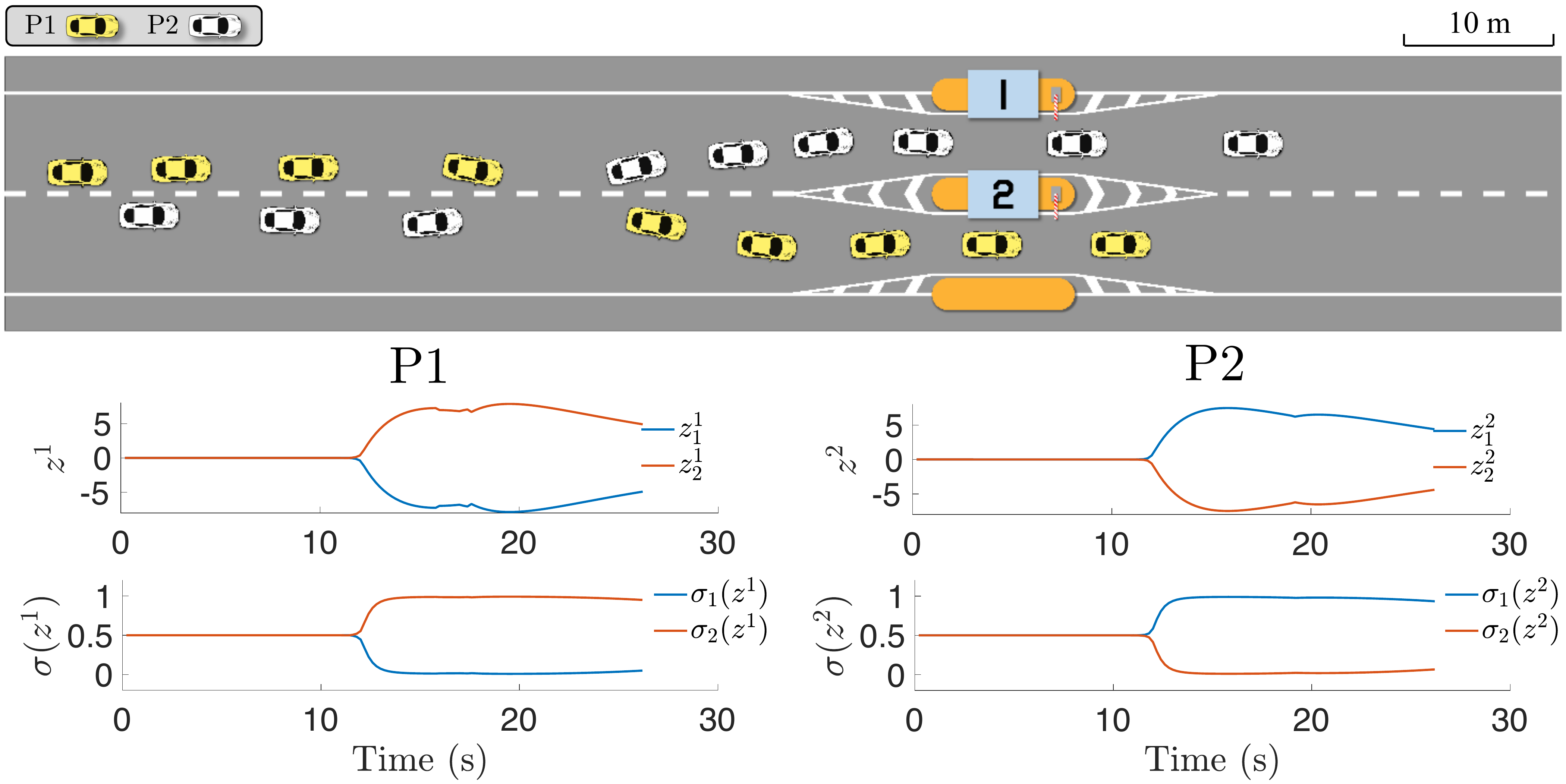}
  \caption{\label{fig:L1L0} State and opinion (original and softmax) trajectories with heterogeneous cost weights. Car 1 uses the L1-QMDP policy and car 2 uses the L0-QMDP policy.
  }
\end{figure}

\section{Conclusions}
\label{sec:conclusions}
We proposed a principled algorithmic approach for synthesizing game-induced nonlinear opinion dynamics (GiNOD) based on the value functions of dynamic games under different agent intent parameters.
In particular, we provided a detailed stability analysis for GiNOD in the two-player two-option case.
Finally, we developed a trajectory optimization algorithm that uses opinions evolved via GiNOD as guidance.
Future works include generalizing the stability analysis to the multi-player multi-option case and demonstrating our approach with hardware robotic systems.

\printbibliography{}

\end{document}